\documentclass{sig-alternate}
\usepackage{framed}
\usepackage{url}

\newtheorem{theorem}{Theorem}[section]
\newtheorem{definition}[theorem]{Definition}

\newtheorem{claim}[theorem]{Claim}

\newtheorem{remark}[theorem]{Remark}
\newtheorem{hypothesis}[theorem]{Hypothesis}

\def\bv{\Delta}

\begin{document}
\title{Revisiting the Examination Hypothesis with Query Specific Position Bias}

\numberofauthors{2}
\author{
%
%
\alignauthor
Sreenivas Gollapudi \\
       \affaddr{Search Labs, Microsoft Research}\\
       \email{sreenig@microsoft.com}
\alignauthor
Rina Panigrahy \\
       \affaddr{Microsoft Research}\\
       \email{rina@microsoft.com}
}

\maketitle 

\begin{abstract} 
Click through rates (CTR) offer useful user feedback that can be used to infer the relevance of search results for queries. However it is not very meaningful to look at the raw click through rate of a search result because the likelihood of a result being clicked depends not only on its relevance but also the position in which it is displayed. One model of the browsing behavior, the {\em Examination Hypothesis} \cite{RDR07,Craswell08,DP08}, states that each position has a certain probability of being examined and is then clicked based on the relevance of the search snippets. This is based on eye tracking studies~\cite{Claypool01, GJG04} which suggest that users are less likely to view results in lower positions. Such a position dependent variation in the probability of examining a document is referred to as {\em position bias}. Our main observation in this study is that the position bias tends to differ with the kind of information the user is looking for.  This makes the position bias {\em query specific}.

In this study, we present a model for analyzing a query specific position bias from the click data and use these biases to derive position independent relevance values of search results. Our model is based on the assumption that for a given query, the positional click through rate of a document is proportional to the product of its relevance and a {\em query specific} position bias. We compare our model with the vanilla examination hypothesis model (EH) on a set of queries obtained from search logs of a commercial search engine.  We also compare it with the User Browsing Model (UBM) \cite{DP08} which extends the cascade model of Craswell et al\cite{Craswell08} by incorporating multiple clicks in a query session.  We show that the our model, although much simpler to implement, consistently outperforms both EH and UBM on well-used measures such as relative error and cross entropy.
\end{abstract}

\section{Introduction} 
Click logs contain valuable user feedback that can be used to infer the relevance of search results for queries (see \cite{Agichtein06, Joachims02, Joachims05} and references within). One important measure is the click through rate of a search result which is the fraction of impressions of that result in clicks. However it is not very meaningful to look at the raw click through rate of a search result because the likelihood of a result being clicked depends not only on its relevance but also the position in which it is displayed. One model of the browsing behavior, the {\em Examination Hypothesis} \cite{RDR07,Craswell08,DP08}, states that each position has a certain probability of being examined and is then clicked based on the relevance of the search snippets. This is based on eye tracking studies~\cite{Claypool01, GJG04} which suggest that users are less likely to view results in higher ranks. Such a position dependent variation in the probability of examining a document is referred to as {\em position bias}. 
 These position bias values can be used to correct the observed click through rates at different positions to obtain a better estimate of the relevance of the document\footnote{We note that click through rate measure need to be combined with other measures like dwell time as the clicks reflect the quality of the snippet rather than the document. Since this study focuses on click through rates, we interchangeably use the term document even though it may refer to the search snippet.}. This raises the question of how one should estimate the effect of the position bias. One method to estimate the position bias is to simply compute the aggregate click through rates in each position for a given query.  Such curves typically show a decreasing click through rate from higher to lower positions except for, in some cases, a small increase at the last position on the result page.  

However, analyzing the click through rate curve aggregated over all queries may not be useful to estimate the position bias as these values may differ with each query.  For example, Broder \cite{Broder02} classified queries into three main categories, {\em viz}, informational, navigational, and transactional. An informational query reflects an intent to acquire some information assumed to be present on one or more web pages. A navigational query, on the other hand, is issued with an immediate intent to reach a particular site.  For example, the query {\tt cnn} probably targets the site \url{http://www.cnn.com} and hence can be deemed navigational. Moreover, the user expects this result to be shown in one of the top positions in the result page. On the other hand, a query like {\tt voice recognition} could be used to target a good collection of sites on the subject and therefore the user is more inclined to more results including those in the lower positions on the page. This behavior would naturally result in a navigational query having a different click through rate curve from an informational query (see Figure~\ref{fig:pb-in}). Further, this suggests that the position bias depends on the query.


\begin{figure}[t]
\label{fig:pb-in}
\centering
\epsfig{file=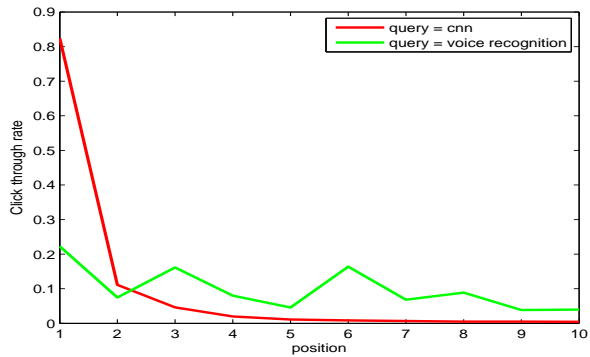, height=2.0in, width=3.5in}
\caption{Click through rate curves over positions $1$ through $10$ for a navigational query, informational  
query. This shows that the click through rate drops differently for different queries
and suggests that the examination probabilities for lower positions may depend on the query.}
\end{figure}

It may be argued that the difference in the click through rate curves for navigational and informational queries arises not from a difference in position bias, but due to the sharper drop in relevance of search results for navigational queries.  In this study, we present a model for analyzing a query specific position bias from the click data and use these biases to derive position independent relevance scores for search results. We note that our model by allowing for the examination to be query specific, subsumes the case of query independent position biases. Our work differs from the earlier works based on Examination Hypothesis in that the position bias parameter is allowed to be query dependent.

\subsection{Contributions of this Study} Our model is based on an
extension of the Examination Hypothesis and states that for a given
query, the click through rate of a document at a particular position
is proportional to the product of its relevance (referred to as
goodness) and query specific position bias.  Based on this model, we
learn the relevance and position bias parameters for different
documents and queries.  We evaluate the accuracy of the predicted CTR by comparing it with the CTR values predicted by the vanilla examination hypothesis and the user browsing model (UBM) of Dupret and Piwowarski~\cite{DP08}.

 We also conduct a cumulative analysis of the derived position bias curves for the different queries and derive a single parametrized equation to capture the general shape of the position bias curve.  The parameter value can then be used to determine the nature of the query as navigational or informational. One of the primary drawbacks of any click-based approach for inferring relevance is the sparsity of the underlying data as a large number of documents are never clicked for a query. We show how to address this issue by inferring the goodness values for unclicked documents through clicks on similar queries.

\section{Related Work}
Several research works have exploited the use of user clicks as feedback in the context of ads and search results. Others have used clicks in conjunction with dwell time and other implicit measures. 

Radlinski and Joachims \cite{Radlinski06} propose a method to learn user preferences for search results by artificially introducing a small amount of randomization in the order of presentation of the results; their idea was to perform flips systematically in the system, until it converges to the correct ranking.  In the context of search advertisements, Richardson {\em et al} \cite{Richardson07} show how to estimate the CTRs for new ads by looking at the number of clicks it receives in different positions.  Similar to our model, they assume the CTR is proportional to the product of the quality of the ad and a position bias.  However, unlike our model, their position bias parameters are query independent. Joachims \cite{Joachims02} demonstrates how click logs can be used to produce training data in optimizing ranking SVMs. 
In another study based on a user behavior, Joachims {\em et al} \cite{Joachims05} suggest several rules for inferring user preferences on search results based on click logs.  For example, one rule `{\tt CLICK > SKIP ABOVE}' means if a user has skipped several search results and then clicks on a later result, this should be interpreted as the user preference for the clicked document is greater than for those skipped above it.  Agichtein {\em et al} \cite{Agichtein06} show how to combine click information based on similar rules along with other user behavior parameters such as dwell time and search result information such as snippets to predict user preferences. Our model, on the other hand, incorporates the CTR values into a system of linear equations to obtain relevance and position bias parameters.  Fox {\em et al} \cite{Fox05} study the relationship between implicit measures such as clicks, dwell time and explicit human judgments.  Craswell {\em et al} \cite{Craswell08} evaluate several models for explaining the effect of position bias on click through rate including one where the click through rate is proportional to the product of relevance and {\em query independent} position bias.  They also propose a cascade model where the click through rate of a document at a position is discounted based on the presence of relevant documents in higher positions. Dupret and Piwowarski~\cite{DP08} present a variant of the cascade model to predict user clicks in the presence of position bias. Specifically, their model estimates the probability of examination of a document given the rank of the document and the rank of the last clicked document.  Guo et al.~\cite{GLK+09} propose a click chain model which is based on the assumption that a document in position $i$ is examined depending on the relevance of the document in the position $i-1$. We will briefly describe these click models next.  Before we do so, we will note the main difference in our work from the earlier works based on the Examination Hypothesis and the Cascade Models is that the position bias parameter is allowed to be query dependent. 

\subsection{Current Click Models}

Two important click models which have been later extended in many works on click models are the {\em examination hypothesis}~\cite{Richardson07} and the {\em cascade model}~\cite{Craswell08}. 

{\em Examination Hypothesis}: Richardson~\cite{Richardson07} proposed this model based on the simple assumption that clicks on documents in different positions are only dependent on the relevance of the document and the likelihood of examining a document in that position.  They assume that the probability of examining the a document at a position depends only on the position and independent of the query and the document. Thus $c_q(d,j) = g_q(d)p(j)$, where $p(j)$ is the position bias of position $j$.

{\em Cascade Model}: This model, proposed by Craswell et al~\cite{Craswell08}, assumes that the user examines the search results top down and clicks when he finds a relevant document. The probability of clicking depends on the relevance of the  This model also assumes that the user stops scanning documents after the first click in the query session. Thus,  the probability of a document $d$ getting clicked in position $j$ is $c_q(d,j) = g_q(d)\prod_{k=1}^{j-1}(1 - g_q(d_k))$ where $d_k$ is the document at rank $k$ in the order presented to the user.In some extensions to this model, other models have considered multiple clicks in a query session.  

{\em Dependent Click Model:} The dependent click model (DCM) proposed by Guo et al~\cite{GLW09} generalizes the Cascade Model to multiple clicks. Once a document has been clicked the next position $j$ may be examined with probability $\gamma(j)$. Thus in a user session if $C_j$ is a binary variable indicating the presence of a click at document $d_j$ at position $j$ then
$Pr(C_{j}=1 | C_{j-1}=1) = g_q(d)\gamma(j)$ and $Pr(C_{j}=1 | C_{j-1}=0) = g_q(d_j)$.

{\em User Browsing Model}: The user browsing model (UBM) proposed by Dupret and Piwowarski~\cite{DP08} is a variant of DCM where the examination parameter $\gamma$  depends not only on the current position but also on the position of the last click.  They assume that given a sequence of click observations in positions $C_{1:j-1}$ the probability of examining position $j$ depends on the position $j$ and the distance $l$ of the position $j$ from the last clicked document ($j=0$ if no document is clicked) and is given by a parameter $\gamma(j,l)$ that is independent of the query. Thus $Pr(C_{j}=1|C_{1:j-1}) = g_q(d_j)\gamma(j,l)$ where $l$ is the distance to the last clicked position. 

{\em Click Chain Model}: This model, due to Guo et al.\cite{GLK+09} is a generalization of DCM that uses Bayesian inference to infer the posterior distribution of document relevance. Here if there is a click on the previous position, the probability that the next document is examined depends on the relevance of the document on the previous position.

Our model is a simple variant of the Examination Hypothesis where the position bias parameter $p(j)$ is allowed to depend on the query $q$. Unlike most prior works out model allows for query specific position bias parameters.

\section{Preliminaries and Model} 
This study is based on the analysis of click logs of a commercial search engine. Such logs typically capture information like the most relevant results returned for a given query and the associated click information for a given set of returned results. In the specific logs that we analyze, each entry in the log has the following form - a query $q$, the top $k$ (typically equal to $10$) documents {\cal D}, the click position $j$, the clicked document $d \in {\cal D}$. Such click data can be be used to obtain the aggregate number of clicks $a_q(d,j)$ on $d$ in position $j$ and the number of impressions of document $d \in {\cal D}$ in position $j$, denoted by $m_q(d,j)$, by a simple aggregation over all records for the given query. The ratio $a_q(d,j)/m_q(d,j)$ gives us the click through rate of document $d$ in position $j$. 

Our study extends the Examination Hypothesis (also referred to as the {\em Separability Hypothesis}) proposed by Richardson {\em et al}~\cite{RDR07} for ads and later used in the context of search results~\cite{Craswell08,DP08}. The examination hypothesis states that there is a position dependent probability of examining a result.  Basically, this hypothesis states that for a given query $q$, the probability of clicking on a document $d$ in position $j$ is dependent on the probability of examining the document in the given position, $e_q(d,j)$ and the relevance of the document to the given query, $g_q(d)$.  It can be stated as

\begin{equation}
\label{eq:eh}
c_q(d,j) = e_q(d,j)g_q(d)
\end{equation}
where $c_q(d,j)$ is the probability that an impression of document $d$ at position $j$ is clicked.  All prior works based on this hypothesis assume that $e_q(d,j) = p(j)$ and depends only on the position and independent of the query and the document. Note that $c_q(d,j)$ can also be viewed as the {\em click through rate} on a document $d$ in position $j$. Thus $c_q(d,j)$ can be estimated from the click logs as $c_q(d,j) = a_q(d,j)/m_q(d,j)$.  We define the {\em position bias}, $p_q(d,j)$, as the ratio of the probability of examining a document in position $j$ to the probability at position $1$.

\begin{definition}[\sc Position Bias]
For a given query $q$, the position bias for a document $d$ at position $j$ is defined as $p_q(d,j) = e_q(d,j)/e_q(d,1)$.
\end{definition}

Next we define the goodness of a search result $d$ for a query $q$ as follows.

\begin{definition}[\sc Goodness]
We define the goodness (relevance) of document $d$, denoted by $g_q(d)$, to be the probability that document $d$ is clicked when shown in position $1$ for query $q$, i.e., $g_q(d) = c_q(d,1)$.
\end{definition}
\begin{remark}
Note that our definition of goodness only seems to measure the relevance of the search result snippet rather than the
relevance of the document $d$. Although this merely a simplification in this study, ideally one needs to combine click through information with other user behavior such as dwell time to capture the relevance of the document.
\end{remark}

The above definition of goodness removes the effect of the position from the click through rate of a document(snippet) and reflects the true relevance of a document that is independent of the position at which it is shown. Having defined the important concepts in our study, we will now state the basic assumption on which our model is based.

\begin{hypothesis}[Document Independence]
\label{hyp:independence}
The position bias $p_q(d,j)$ depends only on the position $j$ and query $q$ and is independent of the document $d$.
\end{hypothesis}

Therefore, we will drop the dependence on $d$ and denote the bias at position $j$ as $p_q(j)$. 
Furthermore, by definition, $p_{q}(1) = 1$ and each entry in the query
log will give us the equation 
\begin{equation}
\label{eq:modbasic}
c_q(d,j) = g_{q}(d)p_{q}(j).
\end{equation}

For a fixed query $q$, we will implicitly drop the $q$ from the subscript for
convenience and use $c(d,j) = g(d)p(j)$.

We note that similar models based on product of relevance and position bias have been used in prior work \cite{Radlinski06, Craswell08}.  However, the main difference in our work is that the position bias parameter $p(j)$ is allowed to depend on the query whereas earlier works assumed them to be global constants independent of the query.

\section{Learning the goodness and position bias parameters}
\label{sec:reg}

In this section we show how to compute the values $g(d)$ and $p(j)$ for a given query based on the above model. Note that different document, position pairs in the click log associated with a given query give us a system of equations $c(d,j) = g(d)p(j)$ that can be used to learn the latent variables $g(d)$ and $p(j)$. Note that the number of variables in this system of equations is equal to the number of distinct documents, say $m$, plus the number of distinct positions, say $n$. We may be able to solve these system of equations for the variables as long as the number of equations is at least the number of variables. However, the number of equations may be more than the number of variables in which case the system is over constrained.  In such a case, we can solve for $g(d)$ and $p(j)$ in such a way that best fit our equations so as to minimize the cumulative error between the left and the right side of the equations, using some kind of a norm.  One method to measure the error in the fit is to use the $L_2$-norm, i.e., $\parallel c(d,j) - g(d)p(j)\parallel_2$.
However, instead of looking at the absolute difference as stated above, it is more appropriate to look at the percentage difference since the difference between CTR values of $0.4$ and $0.5$ is not the same as the difference between $0.001$ and $0.1001$. The basic equation stated as Equation~\ref{eq:modbasic} can be easily modified as 

\begin{equation}\label{eq:logbasic}
\log{c(d,j)} = \log{g(d)} + \log{p(j)}.
\end{equation}

Let us denote $\log g(d)$, $\log p(j)$, $\log c(d,j))$ by $\hat{g}_d$, $\hat{p}_j$, and $\hat{c}_{dj}$, respectively. Let ${\cal E}$ denote the set of all query, document, position combinations in click log.  We get the following system of equations over the set of entries $E_q \in {\cal E}$ in the click log for a given query.

\begin{eqnarray}
\label{eqn:core}
\forall(d,j) \in E_q \mbox{    } \hat{g}_d + \hat{p}_j & = & \hat{c}_{dj} \\
\hat{p}_{1} & = & 0
\end{eqnarray}

We write this in matrix notation $A x = b$ where $x = (\hat{g}_1, \hat{g}_2, \ldots \hat{g}_m, \hat{p}_1, \hat{p}_2, \ldots, \hat{p}_n)$ represents the goodness values of the $m$ documents and the position biases at all the $n$ positions. We solve for the best fit solution $x$ that minimizes $\parallel Ax-b\parallel_2 = \hat{p}_1^2 + \sum_{(d,j) \in E_q} (\hat{g}_d + \hat{p}_j - \hat{c}_{dj})^2$.  The solution is given by $x = (A'A)^{-1} A'b$.

\subsection{Invertibility of $A'A$ and graph connectivity} 

Note that finding the best fit solution $x$ requires that $A'A$ be invertible. To understand when $A'A$ is invertible, for a given query we look at the bipartite graph $B$ (see Figure~\ref{fig:conn-comp}) with the $m$ documents $d$ on left side and the $n$ positions $j$ on the right side, and place an edge if the document $d$ has appeared in position $j$ which means that there is an equation corresponding to $\hat{g}_d$ and $\hat{p}_j$ in Equations~\ref{eqn:core}. We are essentially deducing $\hat{g}_d$ and $\hat{p}_j$ values by looking at paths in this bipartite graph that connect different positions and documents. But if the graph is disconnected we cannot compare documents or positions in different connected components. Indeed we show that if this graph is disconnected then $A'A$ is not invertible and vice versa.

\begin{figure}[t]
\centering
\epsfig{file=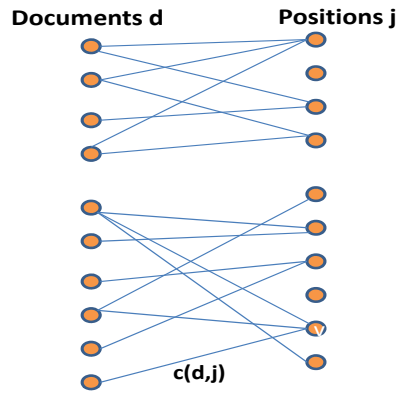, height=2.0in, width=2.0in}
\caption{The bipartite graph $B$ of documents and positions for a given query. The matrix $A$ is invertible only if this graph is connected.  Even if it is disconnected, our model can be used within each connected component.}
\label{fig:conn-comp}
\end{figure}

\begin{claim} 
$A'A$ is invertible if and only if the underlying graph $B$
is connected.  
\end{claim} 
\begin{proof} 
If the graph is connected, $A$ is full rank.  This is because, since $\hat{p}_1 =1$, we can solve for all $\hat{g}_d$ for all documents that are adjacent to position $1$ in graph $B$.  Further, whenever we have a value for a node, we can derive the values of all its neighbors in $B$.  Since the graph is connected, every node is reachable from position $1$. So $A$ has full rank implying that $A'A$ is full ranked and therefore invertible.

If the graph is disconnected, consider any component which does not contain position $1$.  We will argue that the system of equations for this component is not full rank.  This is $Ax = Ax'$ for a solution vector $x$ with certain $\hat{g}_d$ and $\hat{p}_j$ values for nodes in the component, and the solution vector $x'$ with values $\hat{g}_d - \alpha$ and $\hat{p}_j + \alpha$, for any $\alpha$.  Therefore, $A$ is not full rank  as we can have many solutions with same left hand side, implying $A'A$ is not invertible.
\end{proof}

\subsection{Handling disconnected graphs}

Even if the graph $B$ is disconnected, we can still use the system of equations to compare the goodness and position bias values within one connected component. This is achieved by measuring position bias values relative to the highest position within the component instead of position $1$. To overcome the problem of disconnected graphs, we solve for the solution that assumes that the average goodness in the different connected components are about equal. This is achieved by adding the following equations to our system:

\begin{eqnarray}
\label{eqn:core1}
\forall(d) \in E_q \mbox{    } \epsilon(\hat{g}_d - \mu) & = & 0 
\end{eqnarray}

where $\mu$ is the average goodness of the documents for the query and $\epsilon$ is a small constant that tends to $0$. $\epsilon$ simply gives  a tiny weight to these system of equations that is essentially saying that the goodness of all the documents are equal (to $\mu$). If the bipartite graph is connected, these additional equations make no difference to the solution as $\epsilon$ tends to $0$. If the graph is disconnected, it combines the solutions in each connected component in such a way as to make the average goodness in all the components as equal as possible.
 
\subsection{Limitations of the Model}
\label{sec:criticism}
We briefly describe some concerns that arise from our model and describe methods to address some of these concerns.

\begin{itemize}

\item{} The Document Independence Hypothesis ~\ref{hyp:independence} may not be true
   as people may not examine lower positions depending on whether they
   have already seen a good result.  Or they may not click on the next
   document if it is similar to previous one. We show a method to
   measure the extent of validity of this Hypothesis in
   Section~\ref{sec:testing}.

\item{} Some of the connected components of the bipartite graph may be small
 if a limited amount of click data available. 

\item{} For any click based method the coverage of rated documents is small as only clicked docs  
can be rated. In Section~\ref{sec:sim} we show how to increase
coverage by inferring goodness values for unclicked documents through
clicks on similar queries.

\end{itemize}

\section{Experimental Evaluation}
\label{sec:expts}
In this section we analyze the relevance and position bias values obtained by running our algorithm on a commercial search engine click data. Specifically, we adopt widely-used measures such as relative error and perplexity to measure the performance of our click prediction model. Throughout this section, we will refer to our algorithm by {\tt QSEH}, the vanilla examination hypothesis by {\tt EH}, and the user browsing model by {\tt UBM}.  The UBM model was implemented using Infer.Net~\cite{InferNet09}. We show that the our model, although much simpler to implement, outperforms {\tt EH}, and {\tt UBM}.

\subsection*{Click data}
We consider a click log of a commercial search engine containing queries with frequencies between $1000$ and $100000$ over a period of one month.  We only considered entries in the log where the number of impressions for a document in a top-$10$ position is at least $100$ and the number of clicks is non-zero. The truncation is done in order to ensure the $c_{q}(d,j)$ is a reasonable estimate of the click probability.  The above filtering resulted in a click log, call it $\mathcal{Q}$, containing 2.03 million entries with 128,211 unique queries and 1.3 query million distinct documents.  One important characteristics that affect the performance of our algorithm is the frequency.  Table~\ref{tab:summary} summarizes the distribution of query frequencies and the average size of the largest component in each frequency range.  It largely follows our intuition that the more frequent queries are more likely to have a search result shown in multiple positions resulting in a larger component size.

Out of the total $2.03$ million entries, we sample around $85,000$ {\tt query, url, pos} triples into the test set in such a way that there is at least one entry for each unique query in the log; the triples are biased towards urls with more impressions.  Let us denote the test set by $\mathcal{T}$.  This gives us a training set $\mathcal{S} = \mathcal{Q} \setminus \mathcal{T}$ of around $1.9$ million entries.  

\begin{table}
\label{tab:summary}
\centering
\begin{tabular}{|c|c|c|} \hline
Query Freq & Number of  & Avg Size of \\
 & queries & the largest component \\
\hline
< 5000 & 144254 & 2.22 \\
5001 - 10000 & 1911 & 7.85 \\
10001- 15000 & 420 & 8.33 \\
15001 - 20000 & 192 & 8.51 \\
20001 - 25000 & 74 & 8.47 \\
25001 - 30000 & 29 & 8.44 \\
30001 - 35000 & 20 & 8.55 \\
35001 - 40000 & 6 & 6.83 \\
40001 - 45000 & 3 & 9.00 \\
45001 - 50000 & 3 & 9.00 \\
> 50000 & 1 & 9.00 \\
\hline
\end{tabular}
\caption{Summary of the Click Data}
\end{table}

\subsection*{Clickthrough Rate Prediction}
\label{sec:results}

We compute the relative error between the predicted and observed clickthrough rates for each $(q,d,j)$ triple in the test set $\mathcal{T}$ to measure the performance of our algorithm. We compute the relative error as $|c_q(d,j) - \tilde c_q(d,j)|/c_q(d,j)$, where $\tilde c_q(d,j)$ is the predicted CTR from the model and $c_q(d,j)$ is the actual CTR from the click logs. A good prediction will result in a value closer to zero while a bad prediction will deviate from zero. We present the relative error over all triples in $\mathcal{T}$ as a cumulative distribution function in Figure~\ref{fig:cdf-all}.  Such a plot will illustrate the fraction of queries that fall below a certain relative error. For example, for a relative error of 25\%, {\tt EH} produces 48.3\% queries below this error, {\tt UBM} results in 46.12\% queries, while {\tt QSEH} results in 51.57\% queries below this error - an improvement of 10.6\% over {\tt UBM} and 6.34\% over {\tt EH}.  As we can observe from the figure, while {\tt EH} outperforms {\tt UBM} at smaller errors, the trend reverses at larger errors. In Figure~\ref{fig:cdf-comp}, we present the relative error in a different way keeping the sign of the error. This figure shows that {\tt QSEH} does much better in not over predicting the CTRs when compared to {\tt EH} and {\tt UBM} while it does marginally better than {\tt UBM} when it comes to under-prediction. {\tt QSEH} under predicts by an average 48.6\%, {\tt EH} under-predicts by an average 86.54\%, and {\tt UBM} under-predicts by an average 78.09\%.  The respective number in the case of over-prediction are 44.07\%, 78.00\%, and 48.95\%.

\begin{figure}[t]
\centering
\epsfig{file=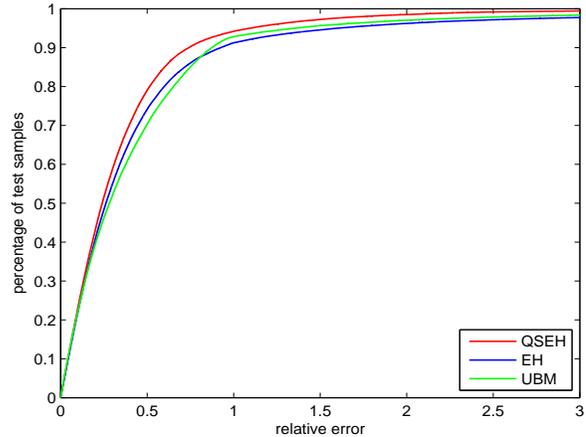, height=2.5in, width=3.5in}
\caption{Cumulative distribution function of the relative error between the predicted and observed CTR values for {\tt QSEH}, {\tt EH}, and {\tt UBM}. For a relative error of $25\%$, {\tt QSEH} outperforms {\tt UBM} by 10.6\% and {\tt EH} by 6.34\%.}
\label{fig:cdf-all}
\end{figure}

\begin{figure}[t]
\centering
\epsfig{file=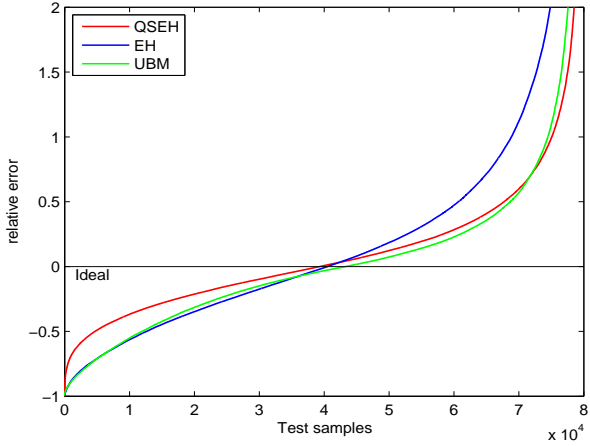, height=2.5in, width=3.5in}
\caption{Relative error for the test data for the {\tt Query Specific EH}, {\tt EH}, and {\tt UBM} methods.}
\label{fig:cdf-comp}
\end{figure}

In another set of experiments, we repeated the above experiment for queries bucketed according to their frequencies to study the effect of query frequency on the CTR prediction.  In this experiment, we estimate average relative error over all test triples for queries in the frequency bucket. The effect of the query frequency is shown in Figure~\ref{fig:cdf-freq}. As the figure illustrates, in the case of {\tt QSEH}, the relative error is stable across all query frequencies while it is higher for the both {\tt EH} and {\tt UBM}.  We note that the stable trends in the figure are for cases where there are reasonable number of queries in that particular frequency range.  We can attribute the large fluctuation in values for frequency greater than $35000$ to the small number of queries in any of the frequency bucket (see Table~\ref{tab:summary}). Finally, we note that the average relative error for {\tt EH} is 39.33\% and {\tt UBM} is 43\%, while it is significantly lower for {\tt QSEH} at 29.19\%.

\begin{figure}[t]
\centering
\epsfig{file=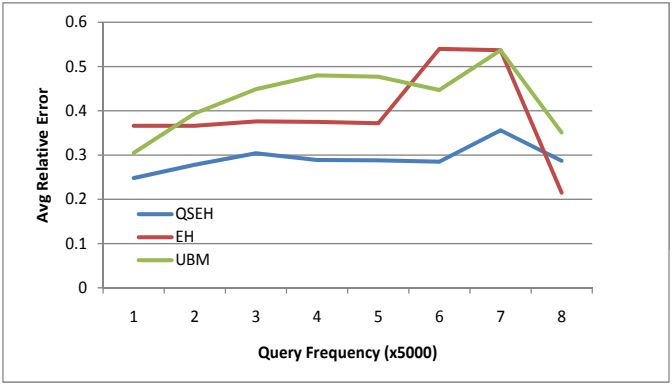, width=2.0in, height=3.5in, angle=-90}
\caption{The relative error between the predicted and observed CTR values for {\tt QSEH}, {\tt EH}, and {\tt UBM} for different query frequencies.  The average relative error for {\tt EH} is 39.33\%, {\tt UBM} is 43\%, and is significantly lower for {\tt QSEH} at 29.19\%.}
\label{fig:cdf-freq}
\end{figure}

Another measure we use to test the effectiveness of our predicted CTRs is perplexity.  We used the standard definition of perplexity for a test set $U$ of query, document, position triples  as \[
P_i = 2^{-\frac{1}{|U|}\sum_{(q,d,j)  \in U} c_q(d,j)\log{\tilde c_q(d,j)}},\] where $c_q(d,j)$ is the observed CTR at position $j$ for query $q$ and document $d$ and $\tilde c_q(d,j)$ is the predicted CTR. This is essentially an exponential function of the cross entropy.  A small value of perplexity corresponds to a good prediction and can be no smaller than 1.0.  We computed the perplexity as a function of the position as well as the query frequency. In the former we group entries in $\mathcal{T}$ by position, and  in the latter, we simply group by  all the queries in a certain frequency range.  Figures~\ref{fig:perplexity-freq} and \ref{fig:perplexity-pos} illustrate the relative performance of {\tt QSEH}, {\tt EH}, and {\tt UBM}. For different query frequencies, the average perplexity of {\tt QSEH} is $1.0671$.  The corresponding values for {\tt EH} and {\tt UBM} are $1.0726$ and $1.0693$ respectively.  In the case of different positions, the corresponding values for {\tt QSEH}, {\tt EH}, and {\tt UBM} are $1.1081$, $1.1286$, and $1.1211$ respectively. 

\begin{figure}[t]
\centering
\epsfig{file=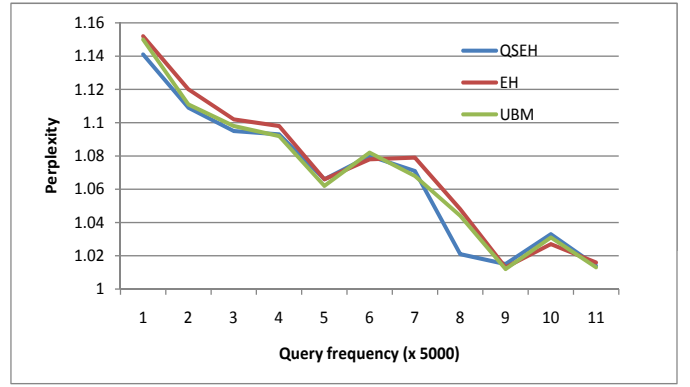, width=2.0in, height=3.5in, angle=-90}
\caption{Perplexity for different query frequencies on the test data. The average perplexity value for {\tt Query Specific EH}, {\tt EH}, and {\tt UBM} is $1.0671$, $1.0726$, and $1.0693$ respectively.}
\label{fig:perplexity-freq}
\end{figure}

\begin{figure}[t]
\centering
\epsfig{file=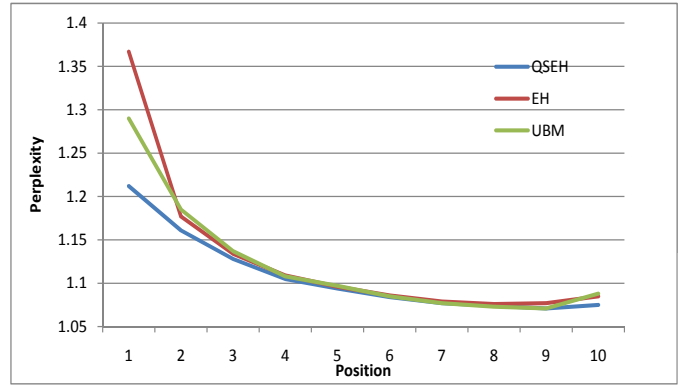, width=2.0in, height=3.5in, angle=-90}
\caption{Perplexity at different positions on the test data. The average perplexity values for {\tt Query Specific EH}, {\tt EH}, and {\tt UBM} are $1.1081$, $1.1286$, and $1.1211$ respectively.}
\label{fig:perplexity-pos}
\end{figure}

\subsection*{Understanding Patterns of position bias}
\label{sec:patposbias}
We also consider a subset of queries, labeled $\mathcal{LC}10$, whose largest component includes all positions 1 through 10 -- these are queries where the bipartite graph B is a fully connected component. This dataset has $112,735$ number of entries with $2,614$ unique queries and $42,119$ unique documents. We use the position bias vectors derived for fully connected components in $\mathcal{LC}10$ to study the trend of the position bias curves over different queries.  A navigational query will have small $p(j)$ values for the lower positions and hence $\hat{p}_j$ ($=\log p(j)$) that are large in magnitude. An informational query on the other hand will have $\hat{p}_j$ values that are smaller in magnitude. For a given position bias vector $p$,  we look at the entropy $H(p) = -\sum_{j=1}^{10}\frac{p(j)}{|p|}\log{\frac{p(j)}{|p|}}$, where $|p|$ is the sum of all the position bias values over all positions. The entropy is likely to be low for navigational queries and high for informational queries. We measured the distribution of $H(p)$ over all the $2500$ queries in $\mathcal{LC}10$ and divided these queries into ten categories of $250$ queries each obtained by sorting the $H(p)$ values in increasing order.

We then study the aggregate behavior of the position bias curves within each of the ten categories. Figure~\ref{fig:bucketpositionbias} shows the median value $\hat{mp}$ of the position bias $\hat{p}$ curves taken over each position over all queries in each category.  Observe that the median curves in the different categories have more or less the same shape but different scale.  It would be interesting to explain all these curves as a single parametrized curve.  To this end, we scale each curve so that the median log position bias $\hat{mp}_6$ at the middle position $6$ is set to $-1$. Essentially we are computing $normalized(\hat{mp}) = -\hat{mp} / \hat{mp}_6$. The $normalized(\hat{mp})$ curves over the ten categories are shown in Figure~\ref{fig:mediannormalpositionbias}. From this figure it is apparent that that the median position bias curves in the ten categories are approximately scaled versions of each other (except for the one in the first category).  The different curves in Figure~\ref{fig:mediannormalpositionbias} can be approximated by a single curve by taking their median; this reads out to the vector $\bv = (0, -0.2952, -0.4935, -0.6792,\\ -0.8673, -1.0000, -1.1100, -1.1939, -1.2284, -1.1818)$.  The aggregate position bias curves in the different categories can be approximated by the parametrized curve $\alpha\bv$.

\begin{figure}[t]
\centering
\epsfig{file=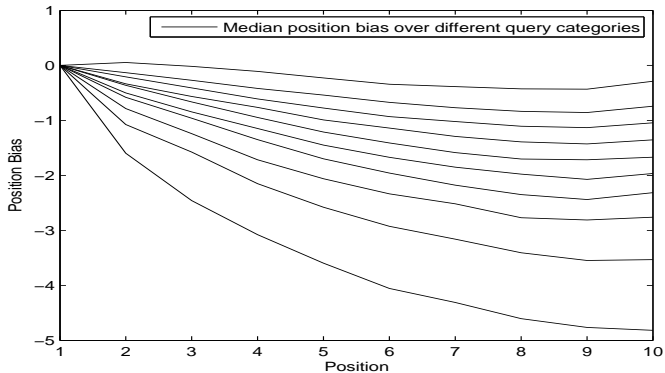, height=2.0in, width=3.5in}
\caption{Position bias curve $\hat{p}_j = \log p(j)$ obtained by taking median
for across $10$ category. The categories are obtained by sorting queries
by sum(p)}
\label{fig:bucketpositionbias}
\end{figure}

\begin{figure}[t]
\centering
\epsfig{file=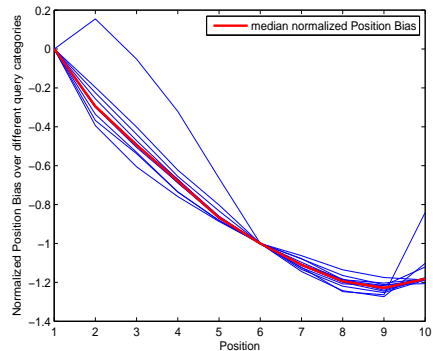, height=2.0in, width=2.5in}
\caption{The best fit curve that describes the general shape of the position bias
curve. This is obtained by taking the median of the normalized position bias curves
over the $10$ categories.}
\label{fig:mediannormalpositionbias}
\end{figure}

Such a parametrized curve can be used to approximate the position bias vector for any query.  The value of $\alpha$ determines the extent to which the query is navigational or information.  Thus the value of $\alpha$ obtained by computing the best fit parameter value that approximates the position bias curve for a query, can be used to classify the query as informational or navigational.  Given a position bias vector $\hat{p}$, the best fit the value of $\alpha$ is obtained by minimizing $\parallel \hat{p} - \alpha \bv\parallel_2$, which results in $\alpha = \bv' \hat{p}/\bv'\bv$. Table~\ref{tab:alphavalues} shows some of the queries in $\mathcal{LC}10$ with the high and low values of $e^{-\alpha}$. Note that $e^{-\alpha}$ corresponds to position bias (since $p(6) = e^{\hat{p}_6}$) at position $6$ as per parametrized curve $\alpha\bv$.

\begin{table}
\label{tab:alphavalues}
\centering
\begin{tabular}{|c|c|c|c|} \hline
 {\tt Query} & {\tt $e^{-\alpha}$}  & {\tt Query} & {\tt $e^{-\alpha}$}  \\
\hline
yahoofinance	& 0.0001  & writing desks	& 0.2919  \\
ziprealty	& 0.0002  & sports injuries	& 0.4250\\
tonight show	& 0.0004 & foreign exchange rates	& 0.7907\\
winzip	& 0.015 & dental insurance	& 0.7944 \\
types of snakes	& 0.1265 &sfo	& 0.8614\\
ram memory	& 0.127 & brain tumor symptoms	& 0.9261 \\
\hline
\end{tabular}
\caption{$e^{-\alpha}$  for a few queries.}
\end{table}

\subsection {Testing the Document Independence Hypothesis}
\label{sec:testing}

Recall that our model is based on the Document Independence
Hypothesis~\ref{hyp:independence}; that is, $p_q(d,j)$ is independent
of $d$. In this section we show a simple method to test this
hypothesis from the click data.

To test the hypothesis we look at the bipartite graph $B$ for a query with documents on one side and positions on the other and each edge $(d,j)$ is labeled by $\hat{c}_{dj}$. We show that cycles in this graph (see Figure~\ref{fig:cyclesinbipartite}) must satisfy a special property.

\begin{figure}[t]
\centering
\epsfig{file=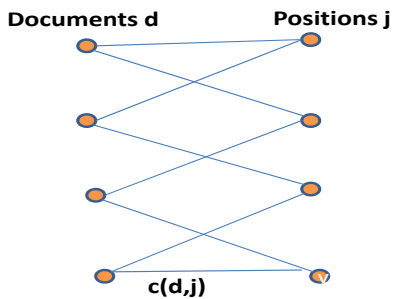, height=1.5in, width=2.0in}
\caption{A cycle $C$ in bipartite graph of documents and positions for a given query. To test the Document Independence Hypothesis
check if $sum(C) = \sum_{(d,j) \in odd\ edges\ of\ C} \hat{c}_{dj} - \sum_{(d,j) \in even\ edges\ of\ C} \hat{c}_{dj} = 0$}
\label{fig:cyclesinbipartite}
\end{figure}

For each edge $(d,j)$ in this graph, we have a $c(d,j)$ obtained from the query log. Let $C=(d_1, j_1, d_2, j_2, d_3, \dots, d_k, j_k, d_1)$ denote a cycle in this graph with alternating edges between documents $d_1, d_2,..,d_k$ and positions $j_1, j_2, .., j_k$ and connecting back at node $d_1$. We now show that our hypothesis implies that the sum of the $\hat{c}_{dj}$ values ($\hat{c}_{dj} = log c(d,j)$) on odd and even edges on the cycle are equal.  This gives us a simple test for our hypothesis by computing the sum for different cycles.

\begin{claim}
\label{clm:testing}
Given a cycle $C = (d_1, j_1, d_2, j_2, d_3, ..., d_k, j_k, d_1)$, our
Independence hypothesis~\ref{hyp:independence} implies that $sum(C) =
\sum_{i=1}^k \hat{c}_{d_ij_i} - \sum_{i=1}^k \hat{c}_{d_{i+1}j_i} = 0$
(where $d_{k+1}$ is the same as $d_1$ for convenience)
\end{claim}

\begin{proof}
We need to show that $\sum_{i=1}^k \hat{c}_{d_ij_i} = \sum_{i=1}^k \hat{c}_{d_{i+1}j_i}$. Note that $\hat{c}_{dj} = \hat{g}_d + \hat{p}_j$.  So $\sum_{i=1}^k \hat{c}_{d_ij_i} = \sum_{i=1}^k \hat{g}_{d_i} + \hat{p}_{j_i}$.  Similarly $\sum_{i=1}^k \hat{c}_{d_{i+1}j_i} = \sum_{i=1}^k \hat{g}_{d_{i+1}} + \hat{p}_{j_i} = \sum_{i=1}^k \hat{g}_{d_i} + \hat{p}_{j_i}$ (since $d_{k+1} = d_1$). 
\end{proof}

Clearly in practice we do not expect $sum(C)$ to be exactly $0$. In fact longer cycles are likely to have a larger error from $0$. To normalize this we consider $ratio(C) = \frac{sum(C)}{\sqrt{ \sum_{i=1}^k |\hat{c}_{d_ij_i}|^2 + \sum_{i=1}^k |\hat{c}_{d_{i+1}j_i}|^2}}$. The denominator is essentially $\parallel C\parallel_2$ where $C$ is viewed as a vector of $\hat{c}_{dj}$ values associated with the edges in the cycle. The number of dimensions of the vector is equal to the length of the cycle.  So $ratio(C) = sum(C)/\parallel C\parallel_2$ is simply normalizing $sum(C)$ by the length of the vector $C$. It can be easily shown theoretically that for a random vector $C$ of length $\parallel C\parallel_2$ in a high dimensional Euclidean space the root mean squared value of $|ratio(C)| = |sum(C)|/\parallel C\parallel_2$ is equal to $1$. Thus, a value of $|ratio(C)|$ much smaller than $1$ indicates that $|sum(C)|$ is biased towards smaller values. This gives us a method to test the validity of the Document Independence Hypothesis by measuring $|sum(C)|$ and $|ratio(C)|$ for different cycles $C$.


We measured the  quantities $|sum(C)|$ and $|ratio(C)|$ computed over different cycles $C$ in the bipartite graphs of documents and positions over different queries. We found a total of $218,143$ cycles of lengths ranging from $4$ to $20$. Note that since this is the bipartite graph the cycle of the smallest length is $4$ and all cycles must be of even length. Figure~\ref{fig:cyclelendistribution} shows the distribution of the length of the different cycles.

For each cycle $C=(d_1, j_1, d_2, j_2, d_3, \dots, d_k, j_k, d_1)$, we compute the quantity $|sum(C)|$ as described in Claim~\ref{clm:testing}. Figure~\ref{fig:cyclessummedian} shows the distribution of $|sum(C)|$. We also plot $|ratio(C)|$ in Figure~\ref{fig:cyclesratiol2median}. 

As can be seen from Figure~\ref{fig:cyclessummedian}, the median value of $|sum(C)|$ is bounded by about $1$ and from Figure~\ref{fig:cyclesratiol2median} the median value of $|ratio(C)|$ is less than $0.1$ for all cycle lengths. While the median value $|sum(C)|$ leaves the validity of the Document Independence Hypothesis inconclusive, the small value of $|ratio(C)|$ can be viewed as mild evidence in favor of the hypothesis.

\begin{figure}[t]
\centering
\epsfig{file=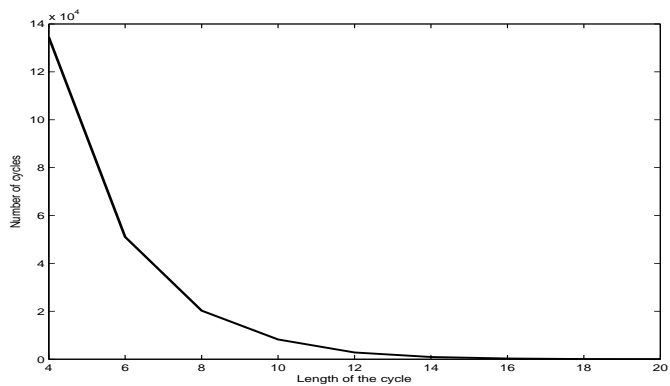, height=2.0in, width=3.5in}
\caption{Distribution of the lengths of the cycles in the bipartite graphs for the queries.}
\label{fig:cyclelendistribution}
\end{figure}

\begin{figure}[t]
\centering
\epsfig{file=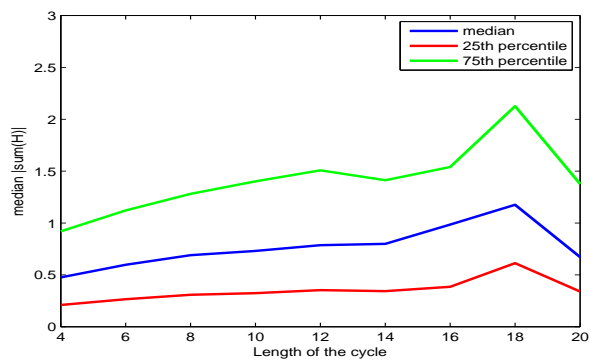, height=2.0in, width=3.5in}
\caption{Distribution of $|sum(C)|$ over cycles of different length.}
\label{fig:cyclessummedian}
\end{figure}

\begin{figure}[t]
\centering
\epsfig{file=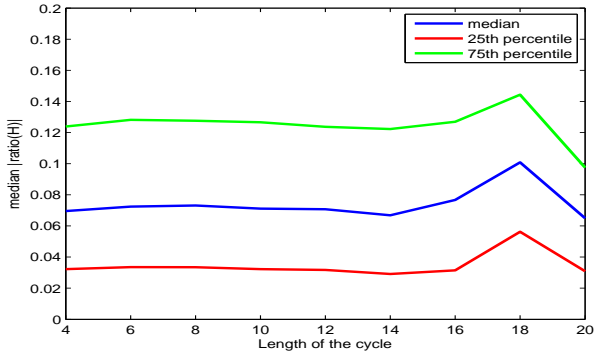, height=2.0in, width=3.5in}
\caption{Distribution of $|ratio(C)|$ over cycles of different length. The fact
that the value of $|ratio(C)|$ is much less than $1$ can be viewed as evidence for
the Document Independence Hypothesis}
\label{fig:cyclesratiol2median}
\end{figure}

\section{Using related queries to increase coverage}
\label{sec:sim}
In addition to finding their use in predicting CTRs, the goodness values obtained from our model
can be employed in designing effective search quality metrics that are very well aligned with user satisfaction.  In this section, we will present a method to infer the goodness values of documents that are not directly associated with a given query (via clicks) and the illustrate the use of these inferred values in computing a click-based feature for ranking search results.

One of the primary drawbacks of any click-based approach is the
sparsity of the underlying data as a large number of documents are
never clicked for a query.  We present a method to extend the
goodness scores for a query to a larger set of documents.  We may be
able to infer the goodness of more documents for a query by looking at
similar queries.  Assuming we have access to a query similarity matrix
$S$, we may infer new goodness values $L_{dq}$ as
\[
L_{dq} = \sum_{q'} S_{qq'}G_{dq'},
\]
where, $S_{qq'}$ denote the similarity between queries $q$ and $q'$.
This is essentially accumulating goodness values from similar queries
by weighting them with their similarity values.  Writing this in
matrix form gives $L = SG$. The question then is how to obtain the
similarity matrix $S$.

One method to compute $S$ is to consider two queries to be similar if
they share a lot of good documents. This can be obtained by taking the
dot product of the goodness vectors spanning the documents for the two
queries.  This operation can be represented in matrix form as $S =
GG'$.  Another way to visualize this is to look at a complete
bipartite graph with queries on the left and documents on the right
with the goodness values on the edges of the graph. $GG'$ is obtained
by first looking at all paths of length 2 between two queries and then
adding up the product of the goodness values on the edges over all the
2-length paths between the queries.

A generalization of this similarity matrix is obtained by looking at
paths of longer length, say $l$ and adding up the product of the
goodness values along such paths between two queries.  This
corresponds to the similarity matrix $S = (GG')^l$.  The new goodness
values based on this similarity matrix is given by $L = (GG')^l G$.
We only use non-zero entries in $L$ as valid ratings.

\subsection*{Relevance Metrics}
\label{sec:metrics}
We measure the effectiveness of our algorithm by comparing the ranking produced when ordering documents for query based on their relevance values to human judgments. We quantify the effectiveness of our ranking algorithm using three well known measures: NDCG, MRR, and MAP. We refer the reader to \cite{Zaragoza04} for an exposition on these measures.  Each of these measures can be computed at different rank thresholds $T$ and are specified by NDCG@T, MAP@T, and MRR@T.  In this study we set $T = {1, 3, 10}$.

The normalized discounted cumulative gains (NDCG) measure\cite{} discounts the contribution of a document to the overall score as the document's rank increases (assuming that the most relevant document has the lowest rank). Higher NDCG values correspond to better correlation with human judgments. Given an ranked result set $Q$, the NDCG at a particular rank threshold $k$ is defined as \cite{Jarvelin02} \[ NDCG(Q,k) = \frac{1}{|Q|}\sum_{j=1}^{|Q|} Z_k\sum_{m=1}^k\frac{2^{r(j)} - 1}{\log(1+j)}, \]  where $r(j)$ is the (human judged) rating (0={\tt bad}, 2={\tt fair}, 3={\tt good},  4={\tt excellent}, and 5= {\tt definitive}) at rank $j$ and $Z_k$ is  normalization factor calculated to make the perfect ranking at $k$ have an NDCG value of 1.

The reciprocal rank (RR) is the inverse of the position of the first relevant document in the ordering.  In the presence of a rank-threshold $T$, this value is $0$ if there is no relevant document in positions below this threshold.  The {\em mean reciprocal rank} (MRR) of a query set is the average reciprocal rank of all queries in the query set.

The average precision of a set of documents is defined as the $\frac{\sum_{i = 1}^nRelevance(i)/i}{\sum_{i=1}^nRelevance(i)}$, where $i$ is the position of the documents in the range $[1,10]$, and $Relevance(i)$ denotes the relevance of the document in position $i$. Typically, we use a binary value for $Relevance(i)$ by setting it to $1$ if the document in position $i$ has a human rating of {\tt fair} or more and $0$ otherwise. The {\em mean average precision}(MAP) of a query set is the mean of the average precisions of all queries in the query set.

\subsection*{Isolated Ranking Features}
One way to test the efficacy of a feature is to measure the effectiveness of the ordering produced by using the feature as a ranking function.  This is done by computing the resulting NDCG of the ordering and comparing with the NDCG values of other ranking features. Two commonly used ranking features in search engines are BM25F~\cite{Robertson93} and PageRank. BM25F is a content-based feature while PageRank is a link-based ranking feature. BM25F is a variant of BM25 that combines the different textual fields of a document, namely title, body and anchor text. This model has been shown to be one of the best-performing web search scoring functions over the last few years \cite{Zaragoza04, Craswell05}. To get a control run, we also include a random ordering of the result set as a ranking and compare the performance of the three ranking features with the control run.

We compute the NDCG scores for this algorithm. We start with a
goodness matrix $G$ 
with $\gamma = 0.6$ containing $936606$ non-zero
entries. Figure~\ref{fig:simndcg} shows the NDCG scores parameter $l$
set to $1$ and $2$ respectively.  The number of non-zero entries
increase to over $7.1$ million for $l = 1$ and over $42$ million for
$l = 2$.  However, the number of judged <query, document> pairs only
increase from $74781$ for $l = 2$ to $87235$ for $l = 1$.  This
implies that most of the documents added by extending to paths of
length $2$ are not judged results in the high value of NDCG scores for
the Random ordering.  If we were to judge all these `holes' in the
ratings, we think that we will see a lower NDCG score for the Random
ordering.

\begin{figure}[t]
\centering
\epsfig{file=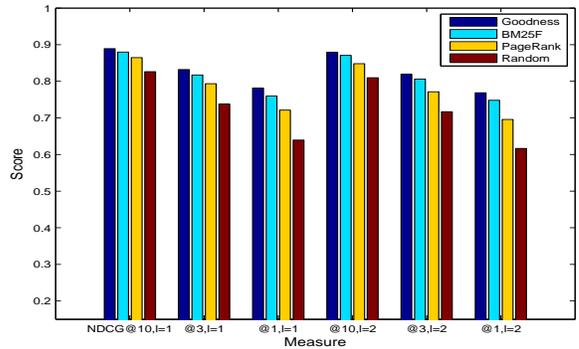, height=2.0in, width=3.5in}
\caption{Since many relevant documents for a query may have no click data we
infer Goodness values by using similar queries.
NDCG scores after thus increasing coverage by  with $l=1$ and $2$ respectively.}
\label{fig:simndcg}
\end{figure}

\section{conclusions}
In this paper, we presented a model based on a generalization of the Examination Hypothesis that states that for a given query, the user click probability on a document in a given position is proportional to the relevance of the document and a query specific position bias. Based on this model we learn the relevance and position bias parameters for different queries and documents.  We do this by translating the model into a system of linear equations that can be solved to obtain the best fit relevance and position bias values. We use the obtained bias curves and the relevance values to predict the CTRs given a query, url, and a position. We measure the performance of our algorithm using well-used metrics like log-likelihood and perplexity and compare the performance with other techniques like the plain examination hypothesis and the user browsing model.

Further, we performed a cumulative analysis of the position bias curves for different queries to understand the nature of these curves for navigational and informational queries. In particular, we computed the position bias parameter values for a large number of queries and found that the magnitude of the position bias parameter value indicates whether the query is informational or navigational.  We also propose a method to solve the problem of dealing with sparse click data by inferring the goodness of unclicked documents for a given query from the clicks associated with similar queries.

\subsection*{Acknowledgements}
The authors would like to thank Anitha Kannan for providing the code to compute the relevance and CTRs using the UBM model. 

\bibliographystyle{plain}
\bibliography{paper}
\end{document}